\documentclass[journal]{IEEEtran}
\usepackage{blindtext}
\usepackage{graphicx}
\usepackage{amssymb}
\usepackage{amsmath}
\usepackage{tikz}
\usepackage{balance}
\usepackage{relsize}
\usetikzlibrary{patterns}
\usepackage{bbm}

\usepackage{pgfplots}

\usepackage{color}

\DeclareMathSymbol{\shortminus}{\mathbin}{AMSa}{"39}

\usepackage{amsfonts}
\usepackage{bm}
\long\def\comment#1{}

\newfont{\bbb}{msbm10 scaled 700}

\newfont{\bb}{msbm10 scaled 1100}

\newcommand{\RR}{\mbox{\bb R}}

\newcommand{\ZZ}{\mbox{\bb Z}}

\newcommand{\EE}{\mathbbm{E}}

\newcommand{\vect}[1]{\lowercase{\mathbf{#1}}}

\newcommand{\random}[1]{\uppercase{#1}}

\newcommand{\randomvect}[1]{\uppercase{#1}}

\newcommand{\ev}{\vect{e}}

\newcommand{\mv}{\vect{m}}

\newcommand{\sv}{\vect{s}}

\newcommand{\wv}{\vect{w}}
\newcommand{\vv}{\vect{v}}
\newcommand{\xv}{\vect{x}}
\newcommand{\yv}{\vect{y}}

\newcommand{\Ac}{{\cal A}}

\newcommand{\Ic}{{\cal I}}

\newcommand{\Pc}{{\cal P}}

\newcommand{\Sc}{{\cal S}}
\newcommand{\Tc}{{\cal T}}

\newcommand{\Wc}{{\cal W}}
\newcommand{\Vc}{{\cal V}}
\newcommand{\Xc}{{\cal X}}
\newcommand{\Yc}{{\cal Y}}

\newcommand{\sr}{\random{s}}

\newcommand{\xr}{\random{x}}
\newcommand{\yr}{\random{y}}

\newcommand{\srv}{\randomvect{s}}

\newcommand{\wrv}{\randomvect{w}}
\newcommand{\vrv}{\randomvect{v}}
\newcommand{\xrv}{\randomvect{x}}
\newcommand{\yrv}{\randomvect{y}}

\newcommand{\deltav}{\bm{\delta}}

\newcommand{\supp}{{\hbox{supp}}}

\newcommand{\eqdef}{\stackrel{\Delta}{=}}

\graphicspath{{./}}

\newcommand{\ceil}[1]{\left\lceil#1\right\rceil}
\newcommand{\floor}[1]{\left\lfloor#1\right\rfloor}
\newcommand{\norm}[1]{\sigma(#1)}

\newcommand{\squeezedequation}{\medmuskip=2mu \thinmuskip=1mu \thickmuskip=3mu}
\newcommand{\supersqueezedequation}{\medmuskip=1mu \thinmuskip=0mu \thickmuskip=2mu \nulldelimiterspace=-1pt \scriptspace=0pt}
\newcommand{\toosqueezedequation}{\medmuskip=0mu \thinmuskip=-1mu \thickmuskip=1mu \nulldelimiterspace=-1pt \scriptspace=0pt}

\newcommand{\bes}{\begin{IEEEeqnarray}{ll}} 
\newcommand{\ees}{\end{IEEEeqnarray}}
\DeclareMathOperator*{\argmin}{argmin}

\usepackage{stmaryrd}
\newcommand{\interval}[2]{\llbracket#1,#2\rrbracket}

\newtheorem{definition}{Definition}
\newtheorem{lemma}{Lemma}
\newtheorem{theorem}{Theorem}

\newenvironment{proof}{
	{\it Proof:}
}{
	{\hfill $\blacksquare$}
}

\begin{document}
%
\title{Universal Mutual Information Privacy Guarantees for Smart Meters}
%
%
%

\author{\IEEEauthorblockN{Miguel Arrieta$^*$, I\~naki Esnaola$^{*\dag}$, and Michelle Effros$^{\S}$}

\IEEEauthorblockA{$^*$Dept. of Automatic  Control and Systems Engineering, University of Sheffield, Sheffield S1 3JD, UK\\
 $^\dag$Dept. of Electrical Engineering, Princeton University, Princeton, NJ 08544, USA\\
 $^\S$ Dept. of Electrical Engineering, California Institute of Technology, Pasadena, CA 91125, USA\\
}
}

%



\maketitle

\begin{abstract}

Smart meters enable improvements in electricity distribution system efficiency at some cost in customer privacy. Users with home batteries can mitigate this privacy loss by applying charging policies that mask their underlying energy use.  A battery charging policy is proposed and shown to provide universal privacy guarantees subject to a constraint on energy cost.  The guarantee bounds our strategy's maximal information leakage from the user to the utility provider under general stochastic models of user energy consumption. The policy construction adapts coding strategies for non-probabilistic permuting channels to this privacy problem.

\end{abstract}



%
\IEEEpeerreviewmaketitle
\vspace{-4mm}
\section{Introduction}
\vspace{-1mm}
Smart meters (SMs) provide advanced monitoring of consumer energy usage, thereby enabling optimized management and control of electricity distribution systems~\cite{ipakchi2009grid}. Unfortunately, the data collected by SMs can reveal information about consumers' activities. 
For instance, an individual's energy usage pattern may leak information about the times at which they run individual appliances~\cite{hart1992nonintrusive}.
Two approaches have been proposed to tackle the privacy threat posed by such information leakage.
One strategy involves manipulating user data before sending it to the utility provider (UP)~\cite{6007070}; this approach improves privacy at the cost of reduced operational insight for the UP.
The other strategy employs rechargeable batteries at each consumer site to try to decouple energy usage from energy requests~\cite{kalogridis2010privacy}; allowing devices to run off of either the battery or the UP and allowing the battery to charge at times of both activity and inactivity improves privacy at the cost of introducing individual batteries and, potentially, increasing consumer costs 
(e.g., if energy is requested when it best conceals the consumers' usage without regard to the energy bill).
This paper investigates the latter approach.

Understanding the privacy implications of any strategy requires an appropriate privacy metric.
A variety of metrics are used to study privacy in energy distribution systems. These include statistical distance metrics~\cite{kalogridis2010privacy}, differential privacy~\cite{6847974}, distortion metrics~\cite{giaconi2018joint}, and information metrics like mutual information, which can be applied under a variety of assumptions on users' energy, including i.i.d.~\cite{varodayan2011smart, kalogridis2010privacy, 6003811, tan2012smart, gomez2013privacy}, stationary~\cite{6102315,65eac443f7d6420a9bb100e3a77b70a6}, and first-order time-homogeneous Markov random processes~\cite{7536745}; see~\cite{GDH_SPM_18} for a comprehensive review.
Alternative privacy metrics such as maximal leakage~\cite{7460507}  have operational descriptions and relate to information measures like Sibson mutual information;  its generalization, maximal $\alpha$-leakage~\cite{LKSP_CORR_18}, establishes additional relationships to Arimoto mutual information, mutual information, and Renyi entropy~\cite{7460507,LKSP_CORR_18}. Many of these measures can be understood as measures of an adversary's ability to gain insight into an unknown random variable $X$ by observing $Y$, 
with measures differing only in the loss functions they use to quantify that insight~\cite{LKSP_CORR_18}.

We here use mutual information to measure privacy both because its interpretation in terms of an adversary that minimizes log-loss with respect to an evolving soft-decision model~\cite{LKSP_CORR_18} is well-matched to the evolving nature of energy distribution over time and because mutual information provides a useful bridge to adjacent fields such as hypothesis testing~\cite{PV_TIT_95}, estimation~\cite{SV_TIT_05}, and learning~\cite{XR_CORR_17}. 

Since user energy consumption may be non-stationary, we seek privacy guarantees that apply across general random process models of energy consumption.  
Moreover, given that
no battery can store unlimited energy, we impose finite capacity bounds on batteries.  We therefore model the energy management unit (EMU) as a deterministic finite-state channel. We then adapt the Ahlswede-Kaspi coding strategy proposed for permuting channels~\cite{ahlswede1987optimal} to the SM privacy setting. This work generalizes the battery policy proposed in \cite{AE_SGC_17} by including the price of the energy requested from the grid and minimizing information leakage subject to a bound on the resulting energy bill.

We denote vectors by bold letters, e.g. $\xv$, and random variables by uppercase letters, e.g.  $X$. The operator $\norm{\cdot}$ denotes the sum over vector elements, e.g. $\norm{\xv}=\sum_i x_i$. Intervals on the integers are denoted by double brackets, e.g. $\interval{a}{b}=\{a, a+1, \ldots, b-1, b\}$. The $n$-fold cartersian product of the interval is denoted by $\interval{a}{b}^n=\interval{a}{b}\times\ldots\times \interval{a}{b}$. Given a vector $\xv$ of size $n$ and a set of indices $\Ac\subseteq\interval{1}{n}$, we denote by $\xv_{\Ac}$ the vector $\xv_{\Ac} = \{x_i:i\in\Ac\}$. The support of the probability distribution $P_X$ is denoted by $\supp(P_X)$, and the positive part operator is ${(a)}^+=\max(0,a)$.


\section{Energy Management System with a Finite Battery Model}

Figure 1 depicts an  energy management system and the random processes therein. The privacy guarantee is defined in terms of the information leakage from the user to the provider, and the task of the EMU is to choose a battery policy that minimizes the leakage while satisfying the operation and cost constraints. Formal definitions follow.

\begin{figure}[t]
	 \centering
	 \includegraphics[width=0.9\columnwidth]{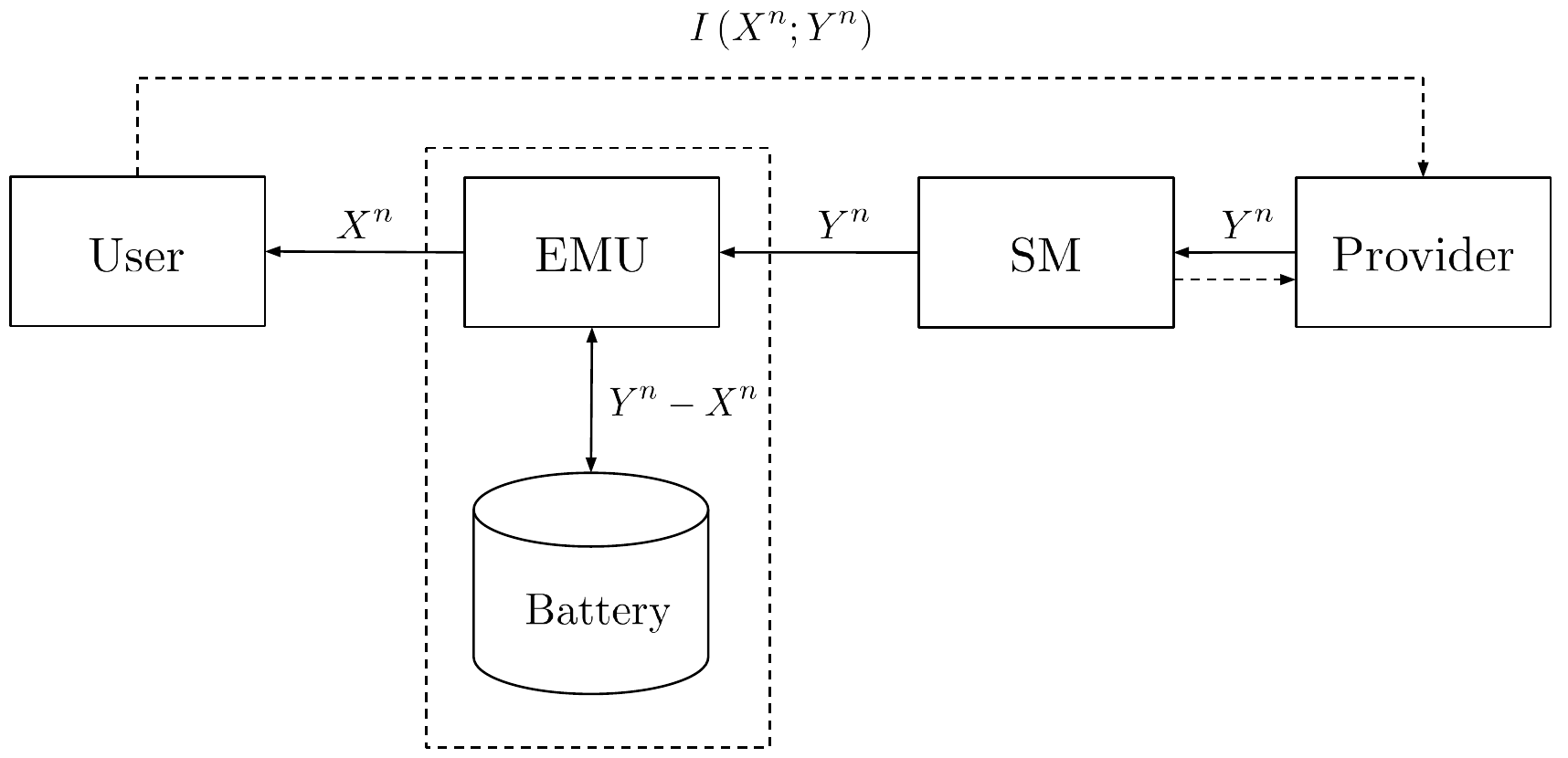}
	 \label{fig:system}
	 \caption{Energy Management System with Finite Battery Model}
\end{figure}

We model user energy consumption as a discrete-time random process $\xrv^n$ on alphabet $\Xc^n = \interval{0}{\alpha}^n$.  The random variable $\xr_i$ describes the energy consumed by the user at time step $i$ with $i=0,1,...,n-1$.  For exposition simplicity we assume $\Xc\subseteq\ZZ$; the results generalize to arbitrary discrete alphabets. We use $P_{\xrv^n}\in\Pc_{\xrv^n}$ to denote the energy consumption pattern distribution, where $\Pc_{\xrv^n}$ is a fixed family of such distributions.
Since user energy consumption profiles tend to exhibit non-stationarities~\cite{kalogridis2010privacy}, 
$\Pc_{\xrv^n}$ may contain non-stationary random processes.

The EMU maps consumption sequence $\xrv^n\in\Xc^n$ to a request sequence $\yrv^n\in\Yc^n$ using a battery policy $P_{\yrv^n|\xrv^n}$ that is not allowed to vary with $\xrv^n$;  random variable $\yr_i$ describes the energy requested from the UP at time step $i=0,1,...,n-1$.  We again focus on integer random variables ($\Yc\subseteq\ZZ$) for simplicity.  We require $\Yc\supseteq\Xc$ so that the UP can satisfy the user's energy consumption even when no battery is available.  We allow $\Yc$ to contain negative values to model scenarios where users can sell energy back to the grid. 

To be considered feasible, battery policy $P_{\yrv^n|\xrv^n}$ must create a request sequence that meets the energy demands of the user and does not request energy it cannot use or store.  Let $\beta$ denote the finite capacity of a given battery (in energy units) and $\sr_i$ denote the amount of energy stored in that battery, the ``energy state,'' at time $i$. Then $\sr_{i}$ takes values in $\Sc=\interval{0}{\beta}$ and is governed by the charging dynamics
\bes \label{eq:battery_filling}
 \sr_{i} = s_0+ \sum_{k=0}^{i-1} \yr_{k} - \sum_{k=0}^{i-1} \xr_{k},
\ees
where $s_0 \in \Sc$ is the initial battery state. 
A power outage occurs when $\sr_{i} + { \yr_{i} - \xr_{i} } < 0$; energy is wasted when $\sr_{i} + { \yr_{i} - \xr_{i} } > \beta$. 
Under this model, the battery resembles a box, energy units resemble balls that can be inserted (stored) and removed (consumed), and the set $\Yc^n(s_0,\xv)$ of feasible requests, defined formally below, contains all sequences of insertions and removals allowed by the box. This feasibility constraint resembles~\cite{ahlswede1987optimal}[Eq. 2.4] from the work of Ahlswede and Kaspi; this link is studied in~\cite{AE_SGC_17}.
\begin{definition} \label{def:feasible_set}
	Given a battery with initial state $s_0\in\Sc$ and capacity $\beta$, the {\it set of feasible energy requests} for energy consumption sequence $\xv\in\Xc^n$ is 
	\bes \label{eq:feasible_energy_requests}
	 \Yc^n(s_0,\xv) \eqdef \{ \yv\in\Yc^n : s_{i} \in \interval{0}{\beta}\ \ \forall\ i\in\interval{0}{n}\}. \IEEEeqnarraynumspace
	\ees
The {\it set of feasible battery policies} is 
	\bes \label{eq:omega_policies}
		\Omega(s_0)\eqdef \{ P_{\yrv^n|\xrv^n} : \supp( P_{\yrv^n|\xrv^n=\xv}) \subseteq \Yc^n(s_0,\,\xv)\ \ \forall \ \xv\in\Xc^n \}. \IEEEeqnarraynumspace \toosqueezedequation
	\ees
\end{definition}

Our aim in feasible policy design is to minimize privacy subject to a constraint on policy cost.  Towards this end, we next define our measures of information leakage (where privacy is high when information leakage is low) and cost.

We measure a battery policy's information leakage by its worst-case performance.


\begin{definition} \label{def:privacy_measure}
	The {\it information leakage} of policy $P_{\yrv^n|\xrv^n}$ is 
	\bes \label{eq:privacy_measure}
		\bar{\Ic}(P_{\yrv^n|\xrv^n})=\max_{P_{\xrv^n}\in\Pc_{\xrv^n}} \frac{1}{n} I(\xrv^n;\yrv^n).
	\ees
\end{definition}

We measure the cost of a policy $P_{\yrv^n|\xrv^n}$ as the difference between the user's energy bill under that policy and the user's energy bill under the feasible battery policy that minimizes the energy bill.  (Under this definition, cost can be negative only for infeasible policies.) To calculate energy bills, we model the energy market price as a deterministic sequence, $\mv\in\RR^n$. Under this definition, the cost of an energy request sequence $\yv$ is  $\mv^T\yv$. We assume that the market price is constant over each of $K$ blocks of time. The price and duration of the $k$-th block, $k=0, 1, \ldots, K-1$, are $m_k$ and $l_k$, respectively, giving 
\bes \label{eq:market_price_model}
	\mv = (\underbrace{ m_0, \ldots,m_0}_{l_0} ,\underbrace{ m_1,\ldots ,m_1}_{l_1}, \ldots, \underbrace{ m_{K-1}, \ldots,m_{K-1}}_{l_{K-1}}). \IEEEeqnarraynumspace
\ees
%
\begin{definition} \label{def:cost}
Consider an EMU with battery capacity $\beta$, initial state $s_0\in\Sc$, and market price $\mv$.  The {\it system cost of energy consumption sequence $\xv\in\Xc^n$ under battery policy $P_{\yrv^n|\xrv^n}$} is 
\bes
	g(\yrv^n, \xv) = \EE_{P_{\yrv^n|\xrv^n=\xv}}[\mv^T\yrv^n - \mv^T\yv^*(\xv)],
\ees
where $\yv^*(\xv) = \argmin_{\yv\in\Yc^n(s_0, \xv)} \mv^T\yv$.  For any $\Delta\geq 0$, the set of {\it feasible $\Delta$-affordable battery policies} is 
	\bes \label{eq:extended_policies}
	 	\Gamma(\Delta) \eqdef \left\{ P_{\yrv^n|\xrv^n}\in\Omega(s_0) : g(\yrv^n, \xv) \leq \Delta\ \ \forall\ \xv\in\Xc^n \right\}\!. \IEEEeqnarraynumspace
	\ees
\end{definition}

Finally, the privacy-cost function defines the optimal tradeoff between privacy and cost achievable by feasible battery policies.

\begin{definition} \label{def:privacy_guarantee}
	Given an EMU with battery capacity $\beta$, initial state $s_0$ and market price $\mv$, the {\it privacy cost function} is defined, for each $\Delta\geq0$, as 
	\bes \label{eq:privacy_guarantee}
		\Ic(\Delta)\eqdef \min_{P_{\yrv^n|\xrv^n}\in\Gamma(\Delta)} \bar{\Ic}(P_{\yrv^n|\xrv^n}).
	\ees
\end{definition}

To bound $\Ic(\Delta)$, we adapt techniques developed by Ahlswede and Kaspi~\cite{ahlswede1987optimal} from channel capacity to privacy-cost. While the resulting solution employs a non-causal battery policy,  detailed analysis of~\cite{ahlswede1987optimal} shows that knowing just $\beta+1$ time steps ahead suffices to achieve optimality, where $\beta$ is the battery capacity. Thus, we envision practical implementations that rely on consumption predictions. This approach also provides insight on what prediction capabilities are needed.

%
%

\section{Geometry of the Feasible Sets}

\subsection{Shared Output Sequences}

\newcommand{\yva}{\yv_{\!\Ac}}
Lemma \ref{le:shared_request_iff} characterizes a necessary and sufficient condition under which a set $\Ac$ of input pairs $(s_0,\xv)$ share a common feasible output sequence $\yva$.  Such shared output sequences are good for privacy since a UP that sees $\yva$ cannot distinguish which input pair $(s_0,\xv)\in\Ac$ caused it. Conversly, when two inputs $(s_0,\xv), (\hat{s}_0,\hat{\xv})$ share no feasible output $\yva$, the EMU cannot hide from the UP which pair caused the request. The following measure of distance 
is useful for that analysis.
\begin{definition}
The distance between two input pairs $(s_0,\xv), (\hat{s}_0,\hat{\xv})\in\Sc\times\Xc^n$ is defined as
\bes
		d_n\Big( (s_0,\xv), (\hat{s}_0,\hat{\xv}) \Big) = \!\!\!\max_{i\in\interval{0}{n-1}} \Big| \left(s_0 - \norm{\xv^i} \right)- \left( \hat{s}_0 - \norm{\hat{\xv}^i} \right) \Big|.\! \IEEEeqnarraynumspace \supersqueezedequation
	\ees
\end{definition}

Lemma \ref{le:shared_request_iff} shows that the distance between input pairs determines the existence of a shared feasible output $\yv$.
The result emphasizes the central role that battery capacity $\beta$ plays in privacy.

\begin{lemma} \label{le:shared_request_iff}
	Let $\Ac$ denote a subset of the input pair alphabet $\Sc\times\Xc^n$. The following two statements are equivalent.

	a) The distance between every two pairs $(s_0,\xv),(\hat{s}_0,\hat{\xv})\in\Ac$ is less than or equal to the capacity of the battery, i.e.
	\bes \label{eq:shared_request_a}
		d_n\Big( (s_0,\xv), (\hat{s}_0,\hat{\xv}) \Big) \leq \beta \textrm{ for all } (s_0,\xv), (\hat{s}_0,\hat{\xv})\in \Ac. \IEEEeqnarraynumspace	
	\ees
%
	\quad b) All sequences in $\Ac$ share a feasible request $\yva$, i.e. 
	\bes \label{eq:shared_request_b}
	\yva \in \bigcap_{(s_0,\xv)\in \Ac }	\Yc^n(s_0,\xv).
	\ees
\end{lemma}

\begin{proof}
	Let the sequence $\yva$ be such that for all $i$:
	\bes
		\norm{\yva^i} = -\min_{(s_0,\xv)\in \Ac } (s_0 - \norm{\xv^{i}}).
	\ees
	Thus, for any $(\hat{s}_0,\hat{\xv})\in\Ac$, the battery state at time $i+1$ is
	\bes \label{eq:iff_achievability}
		s_{i+1} =  (\hat{s}_0 - \norm{\hat{\xv}^{i}}) - \min_{(s_0,\xv)\in \Ac } (s_0 - \norm{\xv^{i}}). \IEEEeqnarraynumspace \squeezedequation
	\ees
	Now $d_n\big( (s_0,\xv), (\hat{s}_0,\hat{\xv}) \big) \leq \beta$ implies that $s_{i+1} \in \interval{0}{\beta}$ for all $i$, so $\yva$ is a feasible sequence.
	The converse follows since for any sequence $\yv$ and any two input pairs $(s_0,\xv),(\hat{s}_0,\hat{\xv})\in\Ac$ such that $d_n\big( (s_0,\xv), (\hat{s}_0,\hat{\xv}) \big) > \beta$, the absolute difference between the corresponding battery states at some time step $i$ satisfies
	\bes
		\big|s_{i+1}-\hat{s}_{i+1}\big| = \left|(s_0 - \norm{\xv^{i}}) -  (\hat{s}_0 - \norm{\hat{\xv}^{i}})\right|>\beta.
	\ees
	Thus $s_{i+1}$ and $\hat{s}_{i+1}$ cannot both belong to $\Sc=\interval{0}{\beta}$.
\end{proof}

\subsection{Cardinality bounds}

Building on Lemma~\ref{le:shared_request_iff}, Theorem~\ref{th:covering} gives an upper bound on the number of distinguishable input pairs $(s_0,\xv^n)\in\Sc_0\times\Xc^n$, where $\Sc_0\subseteq\Sc$ is the set of possible initial battery states. The result is derived by building a covering $\{\Ac_i\}$ of $\Sc_0\times\Xc$  such that all input pairs in each $\Ac_i$ share a common feasible request $\yv_i$.  The result shows that the minimal time $\lambda\eqdef \floor{ {(\beta+1)}/{\alpha} }$ needed to fully discharge a battery of capacity $\beta$ under maximal consumption $\alpha\eqdef\max\Xc$   is a central parameter in the construction of privacy preserving battery policies.  The proof is inspired by the code construction presented by Ahlswede and Kaspi \cite[Proposition 1]{ahlswede1987optimal}.

\begin{theorem} \label{th:covering}
	Let the input alphabet be $\Sc_0\times\Xc^n$, with $\overline{\Sc_0}$ and $\underline{\Sc_0}$ denoting the maximum and minimum values of $\Sc_0$, respectively. There exists a set of request sequences $\Vc^n(\Sc_0) \subseteq \Yc^{n}$ such that
	\bes
		\log \big|\Vc^n(\Sc_0) \big| \leq \ceil{ \frac{n- \floor{(\beta+\underline{\Sc_0}-\overline{\Sc_0})/\alpha}   }{\lambda} }.
	\ees
	Moreover, for every input pair $(s_0,\xv) \in \Sc_0 \times \Xc^{n}$, at least one sequence $\vv\in\Vc^n(\Sc_0)$ is feasible, that is
	\bes
		\Yc^n(s_0,\xv) \cap \Vc^n(\Sc_0) \not= \emptyset.
	\ees

\end{theorem}

\begin{proof}
	At time step $i$, the value of $s_0 - \norm{\xv^{i}}$ for any input pair $(s_0,\xv)\in\Sc_0\times\Xc^{i}$ with $\Xc = \interval{0}{\alpha}$ is bounded by
	\bes \label{eq:zbound}
		\underline{\Sc_0}- i\alpha \leq s_0 - \norm{\xv^{i}} \leq \overline{\Sc_0}.
	\ees
	At time step $l = \floor{(\beta+\underline{\Sc_0}-\overline{\Sc_0})/\alpha}$, the distance between any two input pairs $(s_0,\xv),(\hat{s}_0,\hat{\xv})\in\Sc_0\times\Xc^{l}$ is bounded by
	\bes
		d_l\Big( (s_0,\xv), (\hat{s}_0,\hat{\xv}) \Big) \leq \overline{\Sc_0} - (\underline{\Sc_0}- l\alpha) \leq \beta.
	\ees
	Therefore, Lemma \ref{le:shared_request_iff} guarantees the existence of a request $\yv_0$ that is feasible for every input pair in $\Sc_0\times\Xc^{l}$.
	Following a similar reasoning, consider the set of possible input pairs during the subsequent $\lambda$ times steps, i.e. $\Sc\times\Xc^\lambda$ with $\Sc = \interval{0}{\beta}$. Define a cover of the input alphabet, $\Sc\times\Xc^\lambda\subseteq \left(\Ac_1 \bigcup \Ac_2\right)$, with subsets given by
	\bes
		\Ac_1 = \left\{ (s_0,\xv) \in \Sc\times\Xc^\lambda : s_0 - \norm{\xv} \in \interval{0}{\beta} \right\}, \IEEEeqnarraynumspace
	\ees
	and
	\bes
		\Ac_2 = \left\{ (s_0,\xv) \in \Sc\times\Xc^\lambda : s_0 - \norm{\xv} \in \interval{-\lambda\alpha}{-1} \right\}. \IEEEeqnarraynumspace
	\ees
	Note $\Ac_1 \bigcup \Ac_2$ contains all sequences in $\Sc\times\Xc^\lambda$ as (\ref{eq:zbound}) implies that $s_0 - \norm{\xv} \in \interval{-\lambda \alpha }{\beta}$.
	The distance between any two input pairs in $\Ac_i$ with $i=1,2$ is bounded by $\beta$. Therefore, by Lemma \ref{le:shared_request_iff}, there exists a shared feasible sequence $\yv_i$ for all pairs in $\Ac_i$. 
	Setting $\kappa = \ceil{(n-l)/\lambda}$ and
	\bes
	\label{eq:req_secs}
		\Vc^n(\Sc_0) = \{\yv_0\} \times \underbrace{ \{\yv_1,\yv_2\} \times ... \times \{\yv_1,\yv_2\} }_{\kappa} \supersqueezedequation \IEEEeqnarraynumspace
	\ees
	completes the proof.
\end{proof}

To map input pairs $(s_0,\xv)$ to energy request in $\Vc^n(\Sc_0)$ it suffices to forecast, at the start of each block of length $\lambda$, whether the battery will deplete during the current block, i.e. $s_0 - \norm{\xv^\lambda} \lessgtr 0$. 
In \cite{AEE_ARXIV_19}, it is shown that  the upper bound in Theorem \ref{th:covering}  is tight.
The construction of the set of request sequences given by (\ref{eq:req_secs}) describes the forecasting capabilities required to implement optimal battery policies.

\begin{theorem} \label{th:packing}
	Let $s_0$ denote the state of the battery at time $0$, and let $\Sc_l$ denote the possible states of the battery at time $l\in\interval{0}{n}$. Then there exists a set $\Wc^l(\{s_0\},\Sc_l) \subseteq \Xc^l$ with cardinality
	\bes \label{eq:packing}
		\left|\Wc^l(\{s_0\},\Sc_l)\right| \geq 2^{\hat{\kappa}} \ceil{\frac{l\alpha-\hat{\kappa} \ceil{\lambda}\alpha}{|\Sc_l|}},
	\ees
	and
	\bes
		\Yc^l(s_0,\wv,\Sc_l) \cap \Yc^l(s_0',\wv',\Sc_l) = \emptyset,
	\ees
	for any distinct $(s_0,\wv)$, $(\hat{s}_0,\hat{\wv})$ in $\Wc^l(\Sc_0,\Sc_l)$, $\lambda = (\beta+1)/\alpha$ and $\hat{\kappa} = \max(0,\floor{l/\ceil{\lambda}-1})$.
\end{theorem}

\begin{proof}
	We prove the result by constructing a set of sequences that satisfies the conditions of the theorem. The construction is done by concatenation of $\hat{\kappa}$ blocks of length $\ceil{\lambda}$ and one block of length $l-\hat{\kappa}\ceil{\lambda}$, i.e.
	\bes
		\Wc^l(\{s_0\},\Sc_l) = \underbrace{ \Wc^{\ceil{\lambda}} \times ... \times \Wc^{\ceil{\lambda}}}_{\hat{\kappa}} \times \Wc^{l - \hat{\kappa}\ceil{\lambda}}_{\Sc_l}. \IEEEeqnarraynumspace
	\ees
	Let	the alphabet defining the first $\hat{\kappa}$ blocks be $\Wc^{\ceil{\lambda}}=\{\wv',\wv''\}$, where $\wv'$ and $\wv''$ are any sequences in $\Xc^{\ceil{\lambda}}$ such that $\norm{\wv'} = 0$ and $\norm{\wv''} = \ceil{\lambda}\alpha$. This implies that
	\bes
		d\Big((s_0,\wv'),(s_0,\wv'')\Big) = |\norm{\wv''}-\norm{\wv'}| = \ceil{\lambda}\alpha > \beta. \IEEEeqnarraynumspace
	\ees
	Therefore, by Lemma \ref{le:shared_request_iff}, no output sequence is shared between the input pairs $(s_0,\wv')$ and $(s_0,\wv'')$, i.e
	\bes
		\Yc^{\ceil{\lambda}}(s_0,\wv')\cap\Yc^{\ceil{\lambda}}(s_0,\wv'') = \emptyset.
	\ees
	Thus, the input sequence $\wv_{\yv}\in\{\wv',\wv''\}$, and the initial battery state of the second block $s_{\ceil{\lambda}} = s_0-\norm{\wv_{\yv}}+\norm{\yv}$ are uniquely determined by the output sequence $\yv$. The argument above can be applied recursively for the first $\hat{\kappa}$ blocks.

	Following a similar reasoning, let the alphabet defining the last block be given by $\Wc^{l-\hat{\kappa}\ceil{\lambda}}=\{\wv_0,\wv_1,...,\wv_N\}$ with $N=\floor{{(l\alpha-\hat{\kappa}\ceil{\lambda}\alpha})/{|\Sc_l|}}$ and $\wv_i$ any sequence in $\Xc^{l-\hat{\kappa}\ceil{\lambda}}$ such that $\norm{\wv_i} = i|\Sc_l|$. Consequently, for any given $\yv$, only one sequence $\wv_j$ satisfies the constraint $s_l = s_{\hat{\kappa}\ceil{\lambda}}-\norm{\wv_j}+\norm{\yv}\in\Sc_l$ simultaneously. This completes the proof.
\end{proof}

\subsection{Impact of the Output Alphabet on Information Leakage}

In the following, we characterize the impact of the output alphabet on the information leakage $\Ic(\Delta)$. In particular, we show that the information leakage does not increase when the policy operates with a constrained output alphabet $\Yc_c$. Lemma \ref{le:outputset} shows that it is possible to remove extreme values, i.e. $y_i \not \in \interval{0}{\alpha}$, while retaining the feasibility of the sequence $\yv\in\Yc(s_0,\xv)$. 
\begin{lemma} \label{le:outputset}
	Let two output alphabets $\Yc_c^n$ and $\Yc^n$ be such that $\interval{0}{\alpha}^n \subseteq \Yc_c^n \subseteq \Yc^n \subseteq \ZZ^n$. Then there exists a function $F_n: \Yc^n \to \Yc_c^n$ such that for any $(s_0,\xv)\in\Sc\times\Xc^n$ and $\yv\in\Yc^n(s_0,\xv)$ it holds that
	\bes
		F_n(\yv) \in \Yc^n_c(s_0,\xv).
	\ees
\end{lemma}
\begin{proof}
	We first define the set of functions $\lbrace h_i\rbrace_{i=1}^n$ that will yield the construction of $F_n$. For each function $h_i$ with $i\in\interval{1}{n}$ set $d_i\in\interval{0}{(y_i-\alpha)^+}$ and define $h_i: \Yc^n \to \Yc_c^n$ as
	\bes
		h_i(\yv) &=
		\begin{cases}
			\yv+d_i(\ev_{i+1}-\ev_{i}),& \text{when }  i\in\interval{1}{n-1}\\
			\yv-d_i\ev_{i},& \text{otherwise}.\\
		\end{cases}
	\ees
	That is, the function $h_i$ reallocates the purchase of $d_i$ units of energy from time step $i$ to the next time step $i+1$. 
	Note that when this occurs on the last time step, i.e. when $i=n$, the excess energy request is not reallocated but removed from the sequence.
	\newcommand\sy{\sv}
	\newcommand\shy{\tilde{\sv}}
	Let $\sy\in\Sc^{n+1}$ be the sequence of battery states induced by the feasible sequence $\yv$. By the battery charging dynamics (\ref{eq:battery_filling}), the sequence of battery states induced by $\tilde{\yv} = h_i(\yv)$ is given by
	\bes
		\shy = \sy-d_i\ev_{i+1},
	\ees
	with $d_i\in\interval{0}{(y_i-\alpha)^+}$. Note that
	\bes
		\shy_{i+1} = \sy_{i+1}-d_i \leq \sy_{i+1},
	\ees
	and since $x_i\leq\alpha\leq y_i-d_i$
	\bes
		\shy_{i+1} = \shy_{i} + (y_i-d_i)-x_i \geq \sy_i.
	\ees
	As $\sy\in\Sc^{n+1}$, this implies that $\shy\in\Sc^{n+1}$, i.e. $h_i(\yv)$ is feasible.
	The above argument shows that any excess energy request can be reallocated to the next time step. A similar argument shows that any excess, i.e. $y_i \geq \alpha$, can be reallocated to the previous time step. Furthermore any excess energy selling, i.e. $y_i<0$, can be reallocated to the next and previous time steps without impacting the feasibility of the energy request. A recursive application of the arguments above yields the existence of the function $F_n$ constructed as 
	\bes
			F_n(\yv)=h_n\circ h_{n-1}\cdots \circ h_1 (\yv),
	\ees
	so that $F_n(\yv) \in \Yc_c(s_0,\xv)$.
\end{proof}

The lemma above shows that battery policies that operate over an output alphabet with a maximum energy request that matches the peak energy consumption of the user, i.e. $\Yc=\interval{0}{\alpha}$, are sufficient to satisfy the feasibility.
\begin{lemma} \label{le:ordering0}
	Let the output alphabet $\Yc$ contain the input interval $\Xc = \interval{0}{\alpha}$, then
	\bes
		\Ic_{\Xc}(\infty) = \Ic_{\Yc}(\infty).
	\ees
\end{lemma}
\begin{proof}
	Lemma \ref{le:outputset} states the existence of a function ${F_n}: \Yc^n \to \Yc_c^n$ such that if $P_{Y^n|X^n}\in\Gamma(\infty)$ then $F_n\circ P_{Y^n|X^n}\in\Gamma\infty)$. The function ${F_n}$ induces the Markov chain 	
	\begin{IEEEeqnarray}{rCl}
		\xrv^n \to \yrv^n \to {F}(\yrv^n).
	\end{IEEEeqnarray}
	Therefore $I(\xrv^n;{F}_n(\yrv^n)) \leq I(\xrv^n;\yrv^n)$ by the data processing inequality. The converse follows by noting that $\Gamma_{\Xc}(\infty) \subseteq \Gamma_{\Yc}(\infty)$. This completes the proof.
\end{proof}

However, in general the function $F_n$ does not preserve the price paid for the energy, as $\yv$ and $F_n(\yv)$ may yield different energy bills. The following lemma identifies the conditions that guarantee that the energy bill do not change after the application of $F_n$.


\begin{lemma} \label{le:outputset2}
	Define output alphabet $\Yc_c^n = \interval{-{\beta}/{\underline{l}}}{{\beta}/{\underline{l}} + \alpha}^n$ where $\underline{l} = \min_k l_k$ and $l_k$ is the length of the $k$-th market price period as defined in (\ref{eq:market_price_model}). Consider a $\Delta$-feasible battery policy $P_{Y^n|X^n}\in\Gamma(\Delta)$. Then there exist a function $\widehat{F}: \Yc^n \to \Yc_c^n$ such that $F\circ P_{Y^n|X^n}\in\Gamma_c(\Delta)$.
\end{lemma}

\begin{proof}
	Note that the battery charging dynamics (\ref{eq:battery_filling}) determine the state of the battery at the market transition times $t_{i+1}=t_k+l_k$ with $i=1, \ldots, k$ as
	\bes
		s_{t_{k+1}}=s_{t_k}-\norm{\xv^{l_k}}+\norm{\yv^{l_k}},
	\ees
	where $s_{t_k}\in\interval{0}{\beta}$ and $\norm{\xv^{l_k}} \in \interval{0}{l_k\alpha}$. Therefore, when $\norm{\yv^{l_k}} \in \interval{-\beta}{\beta+l_k\alpha}$, the battery state $s_{t_{k+1}}$ takes values on $\interval{0}{\beta}$ for any value of the previous state $s_{t_{k}}$. This concludes the proof.
\end{proof}

The lemma above shows that the resulting output sequence $\widehat{F}_n(\yv)$ does not depend on the input pair $(s_0,\xv)$ and instead depends only on the original output sequence $\yv$. This insight leads to the following result.
Lemma \ref{le:ordering} shows that the privacy cost function $\Ic( \Delta)$ does not vary when the EMU operates with a constrained output alphabet $\Yc_c$.
This result is consistent with prior results reported for privacy based on hypothesis testing \cite[Theorem 1]{li2015privacy} and multi-user scenarios \cite[Theorem 2]{gomez2015smart}.

\begin{lemma} \label{le:ordering}
	Define output alphabet $\Yc_c^n = \interval{-{\beta}/{\underline{l}}}{{\beta}/{\underline{l}} + \alpha}^n$ where $\underline{l} = \min_k l_k$ and $l_k$ is the length of the $k$-th market price period as defined in (\ref{eq:market_price_model}). Let $\Ic(\Delta)$ and $\Ic_c(\Delta)$ represent the privacy-cost functions under output alphabets $\Yc^n$ and $\Yc_c^n$ for any output alphabet $\Yc^n\supset\Yc_c^n $. Then 
	\bes
		\Ic_c( \Delta ) = \Ic( \Delta ).
	\ees
\end{lemma}

\begin{proof}
	Let $\Gamma(\Delta)$ and $\Gamma_c(\Delta)$ denote the set of feasible $\Delta$-affordable battery policies under output alphabets $\Yc^n$ and $\Yc_c^n$. It follows from \cite{AEE_ARXIV_19} that a function ${F}: \Yc^n \to \Yc_c^n$ exists such that if $P_{Y^n|X^n}\in \Gamma(\Delta)$ then $F\circ P_{Y^n|X^n}\in \Gamma_c(\Delta)$. Noting that the function ${F}$ induces the Markov chain 	
	\bes
		\xrv^n \to \yrv^n \to {F}(\yrv^n)
	\ees
	yields $I(\xrv^n;{F}_n(\yrv^n)) \leq I(\xrv^n;\yrv^n)$ by the data processing inequality. 
	The converse follows by noting that $\Gamma_c(\Delta) \subseteq \Gamma(\Delta)$.
\end{proof}


We note that the proof for the existence of the function $F$ presented in \cite{AEE_ARXIV_19} requires forecasting of $\underline{l}$ time steps ahead.

\section{Universal Privacy bounds}
In the following, we bound the information leakage given in Definition \ref{def:privacy_guarantee}. We first study the case for which only the feasibility constraint is imposed.

\begin{theorem} \label{th:IOmega}
	The privacy cost function $\Ic(\infty)$ is bounded by
	\bes
		\frac{1}{n} \floor{\frac{n}{\ceil{\lambda}}} \leq \Ic(\infty) \leq \frac{1}{n} \ceil{\frac{n-\floor{\beta/\alpha}}{\lambda}},
	\ees
	where $\lambda=(\beta+1)/\alpha$.
\end{theorem}

\begin{proof}
	\emph{Upper bound.} Theorem \ref{th:covering} shows the existence of a set $\Vc^n(\{s_0\})$ with cardinality bounded by 
	\bes
		\log |\Vc^n(\{s_0\})| \leq \ceil{ \frac{n- \floor{(\beta+s_0-s_0)/\alpha} }{\lambda} } = \ceil{\frac{n-\floor{\beta/\alpha}}{\lambda}}, \IEEEeqnarraynumspace \supersqueezedequation
	\ees
	such that the intersection $\Vc^n(\{s_0\})\cap\Yc(s_0,\xv)$ is not empty for every input pair $(s_0,\xv)$. Letting the output $\yrv^n$ take values in $\Vc^n(\{s_0\})\cap\Yc(s_0,\xv)$ completes the proof.

	\emph{Lower bound.} Theorem \ref{th:packing} shows that there exists a set $\Wc^n = \Wc^n(\{s_0\},\Sc)$ with cardinality bounded by
	\bes
		\log |\Wc^n| \geq \floor{\frac{n}{\ceil{\lambda}}},
	\ees
	such that no two sequences in $\Wc^n$ share a common output sequence, i.e. $H(\wrv^n|\yrv^n) = 0$. Letting $\wrv^n$ take uniformly distributed values over $\Wc^n$ completes the proof.
\end{proof}

Note that for a sampling period $T_0$ and a maximum power consumption $\hat{w} = \alpha/T_0$, the total amount of information leaked during a time interval $T=n T_0$ is bounded by
\bes
	n\Ic(\infty) \leq \floor{\frac{n}{{\lambda}}} = \floor{\frac{T/T_0}{(\beta+1)/(\hat{w}T_0)}} = \floor{\frac{T\hat{w}}{\beta+1}}.
\ees
Thus the upper bound is independent of the sampling period $T_0$, i.e. sampling periods under $T_0=(\beta+1)/\hat{w}$ does not increase the privacy guarantee $\Ic(\infty)$. For integer values of $\lambda$, both bounds on Lemma \ref{th:IOmega} coincide, providing the exact value of the privacy guarantee $n\Ic(\infty)=\floor{{n}/{{\lambda}}}$. Consequently, the step behaviour of the privacy guarantee when $n$ increases, is not an aberration introduced by the tools used in this paper, but the real behaviour of the system. 

Theorem \ref{th:I_non_stingy_user} presents our main result, where we bound the information leakage for arbitrary cost constraints $\Delta$. The proof proceeds by constructing a battery policy that combines two components for every request sequence. One of the components guarantees the feasibility constraint, while the other guarantees the cost constraint.

\newcommand{\maxpx}{\max_{\Pc_{\xrv^n}}}
\newcommand{\sto}{\srv_{\!\,\omega}}
\newcommand{\stg}{\srv_{\!\,\gamma}}
\newcommand{\vno}{\vrv^n_{\!\omega}}
\newcommand{\vng}{\vrv^n_{\!\gamma}}

\newcommand{\stos}{\hat{\srv}_{\!\,\omega}}
\newcommand{\stgs}{\hat{\srv}_{\!\,\gamma}}

\begin{theorem} \label{th:I_non_stingy_user}
	Consider an EMU with battery capacity $\beta$, initial state $s_0$, market price $\mv$, and output alphabet $\Yc^n$ satisfying $\Yc^n_c \subseteq \Yc^n$ with $\Yc^n_c$ defined in Lemma \ref{le:ordering}, then
	\bes
	\label{eq:upper_bound}
		\Ic( \Delta ) \leq \Ic( \infty ) + \Ic_\Gamma(\Delta),
	\ees
	where
	\bes \label{eq:multiletter}
		 \Ic_\Gamma(\Delta)= \min_{P_{\stgs|\stos} \in \Gamma_\omega(\Delta)} \max_{P_{\stos}\in\Pc_{\stos}} \frac{1}{n} I(\stgs-\stos;\stos). \IEEEeqnarraynumspace
	\ees
	Here $\stos$ and $\stgs$ are random processes in $\interval{0}{\beta}^K$ with joint distribution determined by
	\bes \label{eq:gammaprime}
	 	\Gamma_\omega(\Delta) = \left\{ P_{\stgs|\stos} : \EE(\stgs\deltav) \leq \Delta - \beta\norm{(\deltav)^+} \right\}\!, \IEEEeqnarraynumspace
	\ees
	where $ \deltav\in\ZZ^K$ denotes the vector of market price differences, with entries given by $\deltav_0=-m_0$, $\deltav_k = m_{k-1}-m_{k}$ for $k=1, 2, \ldots, K-1$ and $ \deltav_K=m_{K-1}$.
\end{theorem}

\begin{proof}
	We prove the result for $\Yc^n = \ZZ^n$; Lemma \ref{le:ordering} generalizes the proof for every $\Yc^n$ satisfying $\Yc_c^n \subseteq \Yc^n$.
	The proof follows by dividing the optimization process into two steps.
	In the first step, we present a battery policy $\omega$ such that the resulting request sequence $\vrv^n_\omega$ satisfies the power outage and energy waste constraints, i.e., $\omega \in \Omega(s_0)$ as defined in (\ref{eq:omega_policies}).
	These policies are discussed on Theorem \ref{th:IOmega}.
	In the second step, we define a random vector $\vrv^n_\gamma$ such that $\yrv^n = \vrv^n_\omega + \vrv^n_\gamma$ also satisfies the cost constraints.
	Specifically, we set
	\bes
		\vng = \sum_{t\in\Tc} \Big( (\ev_{t} - \ev_{t+1}) (\stg - \sto )_{t} \Big), \label{eq:v_gamma}
	\ees
	where $\Tc$ denotes the ordered set of time steps at which a market transition takes place, i.e., $\Tc=\{~\!0,~\!l_0,~\!l_0\!+\!l_1,\ldots,~\!n\!-\!1\}$.
	This implies that
	\bes
		g(\yrv^n, \xv) &= \EE[ (\stg)_\Tc\deltav +\mv^T\xv -  \mv^T\yv^*(\xv) ] \label{eq:gs}\\
		&= \EE[(\stg)_\Tc\deltav] + \beta\norm{(\deltav)^+}, \label{eq:gs2}
	\ees
	where (\ref{eq:gs}) follows by (\ref{eq:v_gamma}) and the battery charging dynamics (\ref{eq:battery_filling}) and (\ref{eq:gs2}) follow by noting that $\Yc^n = \ZZ^n$. Selecting the transformation $\gamma$ determining $(\stg)_\Tc$ from the set described in (\ref{eq:gammaprime}) yields
	\bes
		I(\xrv&^n;\yrv^n) \leq  I(\xrv^n;\vno)  + I(\xrv^n;\vng|\vno)\\
		&= I(\xrv^n;\vno)  + H(\vng|\vno) - H(\vng|\vno,\xrv^n,\sto) \IEEEeqnarraynumspace  \label{eq:condional}\\
		&= I(\xrv^n;\vno)  + H(\stg-\sto|\vno) - H(\stg-\sto|\sto) \IEEEeqnarraynumspace \label{eq:condional2} \\
		&\leq I(\xrv^n;\vno)  + I(\stg-\sto;\sto),
	\ees
	where (\ref{eq:condional}) follows as $\xrv^n$ and $\vno$ determine $\sto$ by the battery charging dynamics (\ref{eq:battery_filling}); (\ref{eq:condional2}) follows by (\ref{eq:v_gamma}) and noting that $\stg-\sto$ is independent of $\vno$ and $\xrv^n$ given $\sto$. Thus
	\bes
		n&\Ic( \Delta )
		= \!\!\!\min_{P_{\yrv^n|\xrv^n} \in \Gamma(\Delta) } \maxpx I(\xrv^n;\yrv^n)\\
		&\leq \!\!\!\min_{ \gamma \in \Gamma_\omega(\Delta) } \min_{\omega\in\Omega(s_0)} \maxpx \Big(I(\xrv^n;\vno) + I(\stg-\sto;\sto) \Big) \IEEEeqnarraynumspace \\
		%
		%
		%
		&\leq\!\!\! \min_{\omega\in\Omega(s_0)} \max_{\Pc_{{\xrv}^n}}~\! I({\xrv}^n;{\vrv}^n_{\!\omega})  + \!\!\!\min_{ \gamma \in \Gamma_\omega(\Delta) } \max_{\Pc_{\srv_\omega}} ~\! I(\stg-\sto;\sto). \squeezedequation  \IEEEeqnarraynumspace \label{eq:IOmega_IGamma}
	\ees
	This completes the proof.
\end{proof}

While direct computation of the information leakage in (\ref{eq:privacy_guarantee}) relies on finding an $n$-dimensional joint distribution satisfiying $\Gamma(\Delta)$, the bound presented in 
(\ref{eq:upper_bound})
 relies on a $K$-dimensional distribution and the simplified version of $\Gamma(\Delta)$ defined in (\ref{eq:gammaprime}). This significantly eases the computation of the information leakage as described in Section \ref{sec:numerical}.
\newcommand{\Deltamax}{\Delta_{\max}}
Note also that (\ref{eq:multiletter}) implies that $\Ic_\Gamma(0) \leq |\Sc_\omega| = {K}/{n} \log_2 (\beta+1)$ and $\Ic_\Gamma(\Delta)= 0$ for any $\Delta \geq \Deltamax$ with $ \Deltamax = \beta\|\deltav\|_1 -\beta m_0$.
Interestingly, a time-sharing argument  presented in \cite{AEE_ARXIV_19} yields
\bes \label{eq:singleletter}
	\Ic(\Delta) \leq \frac{1}{n}\ceil{\frac{n-\floor{\beta/\alpha}}{\lambda}} + \left(1-\frac{\Delta}{\Deltamax}\right)^{\!\!+} \frac{K}{n} \log_2 (\beta+1).  \IEEEeqnarraynumspace \squeezedequation
\ees

\subsection{Worst case consumption proccess}

\begin{theorem} \label{co:selling}
	Let the output alphabet $\Yc^n$ satisfy $\Yc^n_c\subseteq\Yc^n$ with $\Yc^n_c$ defined in Lemma \ref{le:outputset2}, then the privacy guarantee $\Ic(\Delta)$ as defined by Definition \ref{def:privacy_guarantee}, is bounded by
	\bes
		\Ic(\infty)  + (\Ic_{\Gamma}^{{l'}} - \gamma)^+ \leq \Ic(\Delta),
	\ees
	where ${l_k'} = l_k-\floor{l_k/\ceil{\lambda}-1}^+$, $\Ic_{\Gamma}^{{l'}}$ is defined by Definition \ref{def:market_capacity} and $\gamma = \sum {l_k'}/\lambda$.
\end{theorem}

\begin{proof}
	We prove the result by constructing a random process $\wrv^n$ that achieves the lower bound. Let the input alphabet $\Wc^n\subseteq\Xc^{n}$ be divided according to the market price partitioning, i.e. $\Wc^n=\Wc^{l_1}\times\Wc^{l_2}\times...\times\Wc^{l_K}$. Where each set $\Wc^{l_k}$ is divided in two, i.e. $\Wc^{l_k} = \Wc_{\Omega,k} \times \Wc_{\Gamma}^{l_k-\kappa_k\lambda}$ with $\Wc_{\Omega,k} = {\kappa_k\lambda}$  $\kappa_k = \floor{(l_k+1/\alpha)/{\lambda}-1}^+$. Letting the random processes $\wrv^n_{\Omega}$ and $\wrv_{\Gamma}^n$ take values in $\Wc^n_{\Omega}$ and $\Wc^n_{\Gamma}$ implies that
	\begin{IEEEeqnarray}{ll}
		I(\wrv^n;\yrv^n) &= \sum_{k=1}^K I(\wrv_{\Omega}^{l_k};\yrv^n|\wrv^{l_k})
		+ \sum_{k=1}^K I(\wrv_{\Gamma}^{l_k};\yrv^n|\wrv^{l_k}), \IEEEeqnarraynumspace \supersqueezedequation
	\ees
	by the chain rule. For the first term, it holds that
	\begin{IEEEeqnarray}{ll}
		I(\wrv_{\Omega}^{l_k};\yrv^n|\wrv^{l_k}) &= H(\wrv_{\Omega}^{l_k}|\wrv^{l_k}) - H(\wrv_{\Omega}^{l_k}|\wrv^{l_k},\yrv^n) \IEEEeqnarraynumspace \\
		&= \floor{\frac{l_k}{\ceil{\lambda}}-1}^+,
	\ees
	since by Theorem \ref{th:packing} it holds that $\wrv^{l_k}_{\Omega}$ is uniquely determined by $\wrv^{l_k}$ and $\yrv^n$.
	For the second term, it holds that
	\begin{IEEEeqnarray}{ll}
		\min_{P_{\yrv^n|\xrv^n}\in\Pi(\Yc,\Delta)} \max_{P_{\xrv^n}\in\Pc_{\xrv^n}} \sum_{k=0}^{K} I(\wrv_{\Gamma}^{l_k};\yrv^n|\wrv^{l_k}) &= \Ic_{\Gamma}^{{l_k}}.
	\ees
	This completes the proof.
\end{proof}


\section{Numerical results} \label{sec:numerical}

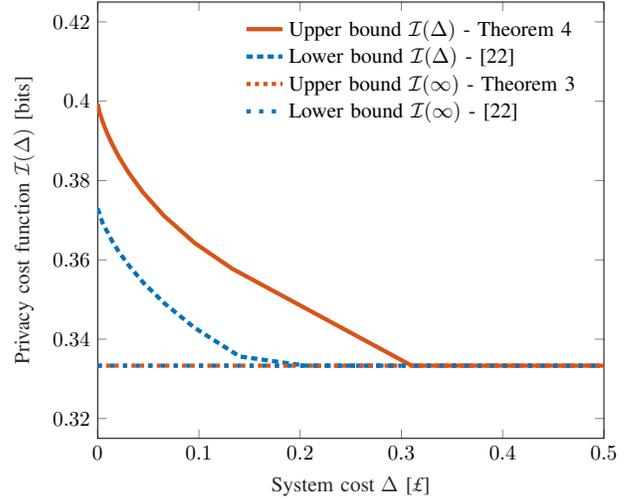
\begin{figure}[t]
	\centering
	\resizebox {0.95\columnwidth} {!} {
%
%
\definecolor{mycolor1}{rgb}{0.00000,0.44700,0.74100}%
\definecolor{mycolor2}{rgb}{0.85000,0.32500,0.09800}%
\begin{tikzpicture}

\begin{axis}[%
xticklabel style={
  /pgf/number format/precision=3,
  /pgf/number format/fixed},
scale only axis,
xmin=0,
xmax=0.5,
xlabel style={font=\color{white!15!black}},
xlabel={$\text{System cost}\;\Delta\text{ [\pounds]}$},
ymin=0.315,
ymax=0.425,
ylabel style={font=\color{white!15!black}},
ylabel={Privacy cost function $\Ic( \Delta )$ [bits]},
axis background/.style={fill=white},
legend style={legend cell align=left, align=left, fill=none, draw=none}
]

\addplot [color=mycolor2, solid, line width=2.0pt]
  table[row sep=crcr]{%
  0	0.399359399971943\\
0.00146265337288939	0.397678295574204\\
0.0021393548892247	0.397053905881144\\
0.0031291346445319	0.396214538367131\\
0.00457683934204961	0.395096496569044\\
0.0066943295008217	0.393624540397733\\
0.00979148362360977	0.391686929298409\\
0.014321546547664	0.389179780324062\\
0.0209474583629331	0.385975725319744\\
0.0306388706280041	0.381955928442367\\
0.0448140474655717	0.377015108570382\\
0.0655474176783434	0.371096216046112\\
0.0958731515514183	0.364250913762259\\
0.132814896985696	0.357791050829485\\
0.3	0.334734870318246\\
0.31	0.333333333333333\\
0.5	0.333333333333333\\
};
\addlegendentry{Upper bound $\Ic( \Delta )$ - Theorem \ref{th:I_non_stingy_user}}

\addplot [color=mycolor1, densely dashed, line width=2.0pt]
  table[row sep=crcr]{%
0	0.373011127013577\\
0.00146265337288939	0.371882962266677\\
0.0021393548892247	0.37125712929101\\
0.0031291346445319	0.370587564809151\\
0.00457683934204961	0.369487507087263\\
0.0066943295008217	0.368367783514192\\
0.00979148362360977	0.366971042795822\\
0.014321546547664	0.364717044927144\\
0.0209474583629331	0.362024859143028\\
0.0306388706280041	0.358639123885903\\
0.0448140474655717	0.354330133368445\\
0.0655474176783434	0.349131291943864\\
0.0958731515514183	0.342843059742629\\
0.140229188486217	0.335652080203089\\
0.205106695516907	0.333333333333333\\
0.3	0.333333333333333\\
};
\addlegendentry{Lower bound $\Ic( \Delta )$ - \cite{AEE_ARXIV_19}}

\addplot [color=mycolor2, dotted, line width=2.0pt]
  table[row sep=crcr]{%
0	0.333333333333333\\
0.5	0.333333333333333\\
};
\addlegendentry{Upper bound $\Ic( \infty )$ - Theorem \ref{th:IOmega} }

\addplot [color=mycolor1, loosely dotted, line width=2.0pt]
  table[row sep=crcr]{%
0	0.333333333333333\\
0.5	0.333333333333333\\
};
\addlegendentry{Lower bound $\Ic( \infty )$ - \cite{AEE_ARXIV_19}}

\end{axis}
\end{tikzpicture}%
	}
	\label{fig:new_bounds}
	\caption{Upper and lower bounds on the privacy cost function as a function of the privacy budget.}
\end{figure}

In this section, we numerically assess the upper bounds on the privacy cost described in Theorem \ref{th:IOmega} and Theorem \ref{th:I_non_stingy_user}. For comparison purposes, we also include the lower bounds on the privacy cost given in \cite{AEE_ARXIV_19}.  We model the market price after the UK Economy 7 tariff, where users are charged an off-peak price of 0.071 \pounds/kWh within a 7 hour block and a peak price of 0.152 \pounds/kWh otherwise \cite{ukEconomy7}. We assume the user has an LG Chem RESU 6.5 battery with a capacity of $4.2$ kWh and a peak power of 4.2 kW. For simplicity we match the users' maximum power consumption to the peak power of the battery, i.e., $4.2$ kW \cite{GDH_SPM_18}. The SM sends the UP integrated energy readings every 30 min following UK specifications for SMs \cite{GDH_SPM_18}. Thus, we set the time elapsed between time steps $i$ and $i+1$ to 30 min. Defining 2.1 kWh as 1 unit of energy yields the following parameters in our system model: battery capacity $\beta= 4.2\;\textnormal{kWh}/2.1\;\textnormal{kWh} = 2$; maximum consumption between time steps $\alpha= 4.2\;\textnormal{kW} \times 0.5\;\textnormal{h}/2.1\;\textnormal{kWh} = 1$; market lengths $l_0 = 7\;\textnormal{h} /0.5\;\textnormal{h} = 14$ and $l_1= 17\;\textnormal{h}/0.5\;\textnormal{h} = 34$; corresponding market prices of $m_0=0.152 \;\pounds/\;\textnormal{kWh} \times 2.1\;\textnormal{kWh} =0.3192$ \pounds\; and $m_1=0.071 \pounds/\;\textnormal{kWh} \times 2.1\;\textnormal{kWh} =  0.1791$ \pounds\; per unit of energy.

Figure 2 depicts the bounds on the privacy cost $\Ic(\Delta)$ for different values of the system cost $\Delta$ and initial battery state $s_0 = 0$ during a one day period, i.e. $n=24\;\textnormal{h}/0.5\;\textnormal{h} = 48$. Following (\ref{eq:singleletter}), when the user does not wish to increase the system cost for privacy, the privacy cost is bounded by $\Ic(0) = 0.4\;\textrm{bits}$. 
For large values of the system cost $\Delta$ 
the cost constraint is always satisfied, i.e.  $\Ic_\Gamma(\Delta)=0$, and the privacy leakage is governed by the feasibility constraints.


\ifCLASSOPTIONcaptionsoff
 
\fi



%




\balance
\bibliographystyle{IEEEtran}
\bibliography{thesisbiblio}

%





\end{document}